\documentclass{article}
 
\usepackage{amssymb,amsmath,amsthm} 
\usepackage{graphicx}
\usepackage[latin1]{inputenc}
\usepackage[a4paper, centering, text={6.5in, 10in}]{geometry} 
\usepackage{authblk}

\newtheorem{theorem}{Theorem}[section]
\newtheorem{lemma}{Lemma}[theorem]

\newtheorem{corollary}{Corollary}[theorem]

\newtheorem{example}{Example}[subsection]
\newtheorem{remark}{Remark}[theorem]

\newcommand{\MZ}{\mathbb Z}
\newcommand{\MR}{\mathbb R}
\newcommand{\MC}{\mathbb C}

\author{Francisco Mota}
\affil{Departamento de Engenharia de Computa\c c\~ao e Automa\c c\~ao\\
Universidade Federal do Rio Grande do Norte -- Brasil\\
e-mail:mota@dca.ufrn.br}

\date{\today}

\title{Computing the Convolution of Analog and Discrete Time Exponential Signals Algebraically}

\begin{document}

\maketitle

\begin{abstract}

We present a procedure for computing the convolution of exponential signals without the 
need of solving integrals or summations. The procedure requires the resolution of a system
of linear equations involving Vandermonde matrices. We apply the method to solve ordinary
differential/difference equations with constant coefficients.
\end{abstract}

\section{Notation and Definitions}
Below we introduce the definitions and notation to be used along the paper:

\begin{itemize}

\item $\MZ$, $\MR$ and $\MC$ are, respectively, the set of integers, real and complex numbers;

\item An analog time signal is defined as a complex valued function 
$f:\underset{t}{\MR}\underset{\mapsto}{\to}\underset{f(t)}{\MC}$, and a discrete time signal 
is a complex valued function $f:\underset{k}{\MZ}\underset{\mapsto}{\to}\underset{f(k)}{\MC}$. 
In this paper we are mainly concerned with exponential signals, that is, $f(t)=e^{rt}$, or $f(k)=r^{k}$, 
where $r\in\MC$. Two basic signals will be necessary in our development, namely, the unit step signal 
($\sigma$) and the unit 
impulse (generalized) signal ($\delta$), both in analog or discrete time setting. The unit step is defined as
\[
\sigma(t) = \begin{cases}0,& t< 0\\ 1, & t>0\end{cases} \quad\text{(analog)}\quad\text{and}\quad
\sigma(k) = \begin{cases}0, & k<0\\ 1, & k\ge 0\end{cases}\quad\text{(discrete time)}
\]
In analog time context we define the unit impulse as $\delta=\dot\sigma$,  where the derivative is 
supposed to be defined in the generalized sense, since $\sigma$ has a ``jump" discontinuity at $t=0$, 
and this is why we denote $\delta$ as a ``generalized" signal \cite{distrib}. 
If $f$ is an analog signal continuous at $t=0$, the product ``$f\sigma$" is given by
\[
(f\sigma)(t) = f(t)\sigma(t) = \begin{cases}0, & t<0 \\
                         f(t), & t> 0
            \end{cases}
\]
and then, if $f(0)\neq 0$, the module of $f\sigma$ also has a ``jump" discontinuity at $t=0$, 
in fact $(f\sigma)(0^-)=0$ while $(f\sigma)(0^+)=f(0)$. Additionally, using the generalized signal 
$\delta$, we can obtain the derivative of $f\sigma$ as
\begin{equation}\label{fsigdot}
\dot{(f\sigma)} = \dot f\sigma + f\dot\sigma = \dot f \sigma + f(0)\delta
\end{equation}
In discrete time context, time shifting is a fundamental operation. We denote by $[f]_n$ the 
shifting of signal $f$ by $n$ units in time, that is, $[f]_n(k) = f(k-n)$. Using this notation, 
the discrete time impulse $\delta$ can be written as $\delta=\sigma-[\sigma]_1$ or 
$\delta(k)=\sigma(k)-\sigma(k-1)$.

\item The convolution between two signals $f$ and $g$, represented by $f*g$,  is the binary
operation defined as \cite{wiki}:
\begin{equation}\label{conv}
\begin{array}{rcl}
(f*g)(t) & = & \displaystyle\int_{-\infty}^{\infty}f(\tau)g(t-\tau)d\tau, \quad\text{for analog signals, or}\\[0.5cm]
(f*g)(k) & = & \displaystyle\sum_{j=-\infty}^{\infty}f(j)g(k-j), \quad\text{for discrete time signals}
\end{array}
\end{equation}
Additionally if we have $f(t)=g(t)=0$ for $t<0$ and $f(k)=g(k)=0$ for $k<0$, we get from (\ref{conv}) that: 
\begin{equation}\label{conv0}
(f*g)(t) = \begin{cases}0, & t<0\\ \displaystyle\int_{0}^{t}f(\tau)g(t-\tau)d\tau, & t>0\end{cases}
\quad\text{and}\quad
(f*g)(k) = \begin{cases}0, & k<0\\ \displaystyle\sum_{j=0}^{k}f(j)g(k-j), & k\ge 0\end{cases}
\end{equation}
Convolution is commutative, associative and the unity of the operation is the signal $\delta$, 
that is, $f*\delta=\delta*f=f$ for any signal $f$. Other important properties of convolution are 
related with derivation and time shifting:
\[
\dot{(f*g)} = \dot f*g = f*\dot g
\]
\[
f*[\delta]_n = [f]_n
\]
\end{itemize}

\section{Introduction}

Convolution between signals is a fundamental operation in 
the theory of linear time invariant (LTI) systems\footnote{In another important context, 
convolution can also be used to compute the probability 
density function of a sum of independent random variables \cite{chung,wikifold}.} 
and its importance comes mainly from the fact that a LTI operator 
$H$, which represents a LTI system in analog or discrete time context, satisfies the following 
property involving signals convolution \cite{ss}:
\begin{equation}\label{ltiop}
H(u*v) = H(u)*v=u*H(v)
\end{equation}
for any signals $u$ and $v$, analog or discrete time defined. Since $u=u*\delta$  for any signal $u$, 
taking in particular $v=\delta$ in (\ref{ltiop}), we get:
\begin{equation}\label{ltiopimp}
H(u)=u*H(\delta)=H(\delta)*u, \quad\text{for any signal } u.
\end{equation}
Equation (\ref{ltiopimp}) above implies that the signal $h=H(\delta)$ (denominated 
{\em impulse response}) characterizes the operator $H$, or the LTI system, in the sense that
the system output due to any input signal $u$, that is $H(u)$, is given by the convolution
between $u$ itself and the system impulse response $h$. This is pretty much similar 
to the fact that a linear function, e.g.\ $f(x)=ax$, is characterized by its value at $x=1$, or 
$f(x)=f(1)x$. 

Maybe the most important class of LTI systems (in analog or discrete time context) are the ones
modeled by a $n$ order ordinary differential/difference equation with constant coefficients, 
as shown bellow:
\begin{equation}\label{edifeq}
\begin{array}{rcl}
y^{(n)}(t) + a_{n-1}y^{(n-1)}(t) + \cdots + a_2\ddot y(t) + a_1\dot y(t) + a_0y(t) & = & u(t), \quad
\text{analog, or} \\[0.3cm]
y(k+n) + a_{n-1}y(k+n-1) + \cdots + a_2y(k+2) + a_1y(k+1) + a_0y(k)  & = & u(k), \quad 
\text{discrete time}
\end{array}
\end{equation}
where $y$ represents the system output signal and $u$ is the system input signal. For this class of
systems, it can be shown that the impulse response $h$ can be written as a convolution between $n$
exponential signals which are defined from the system model (\ref{edifeq}); more specifically \cite{ss}:
\begin{equation}\label{imprespde}
h = h_1*h_2*\cdots*h_n, \quad \begin{cases} h_i(t) = e^{r_it}\sigma(t)& \text{analog, or}\\
                                                                         h_i(k) = r_i^{k-1}\sigma(k-1)& \text{discrete time}
                                                  \end{cases}
\end{equation}
where $r_1,r_2,\ldots, r_n$, $r_i\in\MC$, are the roots of the characteristic equation
$x^n + a_{n-1}x^{n-1} + \cdots + a_2x^2 + a_1x + a_0 = 0$, associated to 
the model (\ref{edifeq}). 

The result in Equation~(\ref{imprespde}) above motivate us to find a procedure to compute 
the convolution between exponential signals. 
In most text books this question is generally dealt in the domain of Laplace or $Z$ transforms, 
where time domain convolution, under certain circumstances, becomes the usual product 
(see e.g. \cite{dh}). In next sections, on the other hand, we show that convolution between exponential
signals can be evaluated directly in time domain, without having to solve integrals or summations, 
just by solving an algebraic system of linear equations involving Vandermonde matrices. This approach 
is very adequate to be implemented computationally in software packages like Scilab \cite{scilab}. 
We should note that, since this is a quite old question,
equivalent results scattered in literature may exists (see e.g. \cite{akk,maliu} for results obtained in the 
context of probability theory); 
but we believe that our approach to this problem is new.
Additionally to find the system impulse response (\ref{imprespde}), we also use the same technique to 
compute the complete solution of the 
differential/difference equation (\ref{edifeq}) for a given signal $u$. 

\section{Convolution between analog exponential signals}\label{analogconv}
Consider the analog time signal $h:\MR\to\MC$ defined by:
\begin{equation}\label{hexp}
h(t) = e^{rt}\sigma(t), \quad r\in\MC, \text{ and }\sigma(t)=\begin{cases}0, & t<0\\1, & t>0\end{cases}
\end{equation}
which is well known to appear as the impulse response of (causal) 
linear time invariant systems  (LTI) modeled by a first order ordinary 
differential equation. We note that $h$, as defined in (\ref{hexp}), has two simple, 
and important, properties:
\begin{enumerate}
\item Its module has a jump discontinuity of amplitude one at $t=0$, more precisely, $h(0^-)=0$ and $h(0^+)=1$;
\item Its derivative ($\dot h$) satisfies the (first order differential) equation $\dot h= rh + \delta$, as can be
deduced from (\ref{fsigdot}).
\end{enumerate}
Now lets consider the convolution between two signals of 
this kind, that is, let be $h_1(t)=e^{r_1t}\sigma(t)$ and $h_2(t)=e^{r_2t}\sigma(t)$. Since both of them 
are zero for $t<0$, we get from (\ref{conv0}) that $(h_1*h_2)(t)=0$ for $t<0$ and, for $t>0$ we have:
\begin{equation}\label{h1*h2}
(h_1*h_2)(t) = \int_0^te^{r_1\tau}e^{r_2(t-\tau)}d\tau = e^{r_2t}\int_0^te^{(r_1-r_2)\tau}d\tau, 
\end{equation}
and, before solving this integral, we note that the convolution $h_1*h_2$ satisfies the properties below:
\begin{enumerate}
\item $h_1*h_2$ is {\em continuous} at $t=0$, more precisely, $(h_1*h_2)(0^-)=(h_1*h_2)(0^+)=0$, since 
by (\ref{h1*h2}) we have:
\begin{equation}\label{h1h20}
(h_1*h_2)(0^+) = \underbrace{e^{r_20^+}}_{=1}\underbrace{\int_0^{0^+}e^{(r_1-r_2)\tau}d\tau}_{=0} = 0;
\end{equation}
and, of course, the integral above is zero because we have an integration of an exponential function over an infinitesimal interval.
\item The derivative of $(h_1*h_2)$, that is $\dot{(h_1*h_2)}$, is such that 
$\dot{(h_1*h_2)}(0^-)=0$ and $\dot{(h_1*h_2)}(0^+)=1$. In fact:
\[
\dot{(h_1*h_2)} = h_1*\dot h_2 = h_1*(r_2h_2+\delta) = r_2(h_1*h_2)+h_1*\delta = r_2(h_1*h_2)+h_1
\]
and then
\begin{equation}\label{h1h2dot0}
\dot{(h_1*h_2)}(0^+) = r_2\underbrace{(h_1*h_2)(0^+)}_{=0}+\underbrace{h_1(0^+)}_{=1} = 1
\end{equation}
\end{enumerate}
Now we return to analyse the integral in (\ref{h1*h2}), by considering two cases:
\begin{enumerate}
\item $r_1\neq r_2$ (or $h_1\neq h_2$):
\begin{eqnarray}
(h_1*h_2)(t) & = & \frac{1}{r_1-r_2}e^{r_1t} \sigma(t)+ \frac{1}{r_2-r_1}e^{r_2t}\sigma(t) \label{r1r2neq}, \text{ or}\\
(h_1*h_2)(t) & = & A_1h_1(t) + A_2h_2(t), \quad A_1=\frac{1}{r_1-r_2}\text{ and } A_2=\frac{1}{r_2-r_1} \label{r1r2neqb}
\end{eqnarray}
\begin{remark}\em
Note that in case where $r_1$ and $r_2$ is a complex conjugate pair, represented by $\alpha\pm j\omega$,
we get from (\ref{r1r2neq}) that $(h_1*h_2)(t) = (e^{\alpha t}/\omega)\sin(\omega t)$ for $t\ge0$.
\end{remark}

From Equation~(\ref{r1r2neqb}) we see that, in case that $r_1\neq r_2$, the convolution $h_1*h_2$ can be written as 
a {\em linear combination} of the signals $h_1$ and $h_2$, and this fact, along with conditions (\ref{h1h20}) and 
(\ref{h1h2dot0}), can be used to find the scalars $A_1$ and $A_2$, without the need of solving the convolution integral
(\ref{h1*h2}), as shown bellow:
\begin{eqnarray*}
(h_1*h_2)(0^+) & = & A_1h_1(0^+) + A_2h_2(0^+)= A_1 + A_2 = 0\\
\dot{(h_1*h_2)}(0^+) & = & A_1\dot h_1(0^+) + A_2\dot h_2(0^+) = A_1r_1 + A_2r_2 = 1
\end{eqnarray*}
or:
\begin{equation}\label{vanderm2}
\begin{bmatrix} 1 & 1\\r_1 & r_2\end{bmatrix}\begin{bmatrix}A_1 \\ A_2\end{bmatrix} = \begin{bmatrix}0 \\ 1 \end{bmatrix}
\implies
\begin{bmatrix}A_1 \\ A_2\end{bmatrix} = \begin{bmatrix} 1 & 1\\r_1 & r_2\end{bmatrix}^{-1}\begin{bmatrix}0 \\ 1 \end{bmatrix}.
\end{equation}
Solving (\ref{vanderm2}) we get $A_1$ and $A_2$ as shown in (\ref{r1r2neqb}).
\item $r_1=r_2=r$ (or $h_1=h_2=h$): 
\begin{equation}\label{r1r2eq}
(h*h)(t) = te^{rt}\sigma(t) = th(t)
\end{equation}

\end{enumerate}

Now we consider a generalization of the results above for a convolution of $n\ge2$ exponential signals 
as shown in (\ref{hexp}).
We start by finding a generalization for the conditions (\ref{h1h20}) and (\ref{h1h2dot0}):

\begin{theorem}\label{conv0+}\em
Consider the convolution of $n\ge2$ signals $\{h_1,h_2,\ldots, h_n\}$ with $h_j(t)=e^{r_jt}\sigma(t)$ and 
$r_j\in\MC$. The $i$-th derivative of $(h_1*h_2*\cdots*h_n)$, represented by $(h_1*h_2*\cdots*h_n)^{(i)}$, evaluated at $t=0^{+}$ is given by:
\[
(h_1*h_2*\cdots*h_n)^{(i)}(0^+) = \begin{cases}0, & i=0,1,\ldots,n-2\\
                                                                    1, & i=n-1
                                                      \end{cases}
\]
and we consider $(h_1*h_2*\cdots*h_n)^{(0)} = h_1*h_2*\cdots*h_n$.                                                                  
\end{theorem}
\begin{proof} We note that $(h_1*h_2*\cdots*h_n)(0^+)=0$ if $n\ge 2$, since this involves an integration of 
exponentials over an infinitesimal interval; this proves that $(h_1*h_2*\cdots*h_n)^{(0)}(0^+)=0$. 
Now consider $(h_1*h_2*\cdots*h_n)^{(i)}$ for $1\le i\le n-2$, then:
\begin{eqnarray}
(h_1*h_2*\cdots*h_n)^{(i)} & = & (\dot h_1*\dot h_2*\cdots*\dot h_{i})*
                                                     \underbrace{(h_{i+1}*\cdots*h_{n-1}*h_n)}_{\text{at least two terms}} \nonumber\\
                                           & = & [(r_1h_1+\delta)*(r_2h_2+\delta)*\cdots*(r_ih_i+\delta)]*(h_{i+1}*\cdots*h_{n-1}*h_n)\nonumber\\
                                           & = & (f+\delta)*(h_{i+1}*\cdots*h_{n-1}*h_n)\nonumber\\
                                           & = & f*(h_{i+1}*\cdots*h_{n-1}*h_n) + (h_{i+1}*\cdots*h_{n-1}*h_n) \label{atl2}
\end{eqnarray}
Since the two terms in (\ref{atl2}) are composed by a convolution of at least two signals, we conclude that 
($h_1*h_2*\cdots*h_n)^{(i)}(0^+)$ is equals to zero. Now, considering $i=n-1$, we have:
\begin{eqnarray}
(h_1*h_2*\cdots*h_n)^{(n-1)} & = & (\dot h_1*\dot h_2*\cdots*\dot h_{n-1})*h_n \nonumber\\
                                               & = & (r_1h_1+\delta)*(r_2h_2+\delta)*\cdots*(r_{n-1}h_{n-1}+\delta)*h_n\nonumber\\
                                               & = & (f+\delta)*h_n = f*h_n + \delta*h_n \nonumber \\
                                               & = & f*h_n + h_n \label{onlyone}
\end{eqnarray}
Then from (\ref{onlyone}), since $f*h_n$ is a sum of (at least) two signals convolution, we have that 
$(f*h_n)(0^+)=0$ and consequently $(h_1*h_2*\cdots*h_n)^{(n-1)}(0^+) = h_n(0^+) = 1$
\end{proof}

In the following we will find a procedure for computing the convolution $h_1*h_2*\cdots*h_n$ for $n\ge 2$ and 
$h_j(t)=e^{r_jt}\sigma(t)$ with $r_j\in\MC$ without the need of solving integrals. 
To begin with, we consider the case where $h_i\neq h_j$ for $i\neq j$, 
which implies $r_i\neq r_j$ for $i\neq j$, and it is just a generalization of (\ref{vanderm2}):

\begin{theorem}\em\label{vandermn}
The convolution between $n\ge 2$ exponentials signals $\{h_1,h_2,\ldots, h_n\}$, with $h_j(t)=e^{r_jt}\sigma(t)$, $r_j\in\MC$ and $h_i\neq h_j$ for $i\neq j$, is given by
\begin{equation}\label{convsum}
h_1*h_2*\cdots*h_n=A_1h_1+A_2h_2 + \cdots + A_nh_n,
\end{equation}
where $A_j\in\MC$ are scalars that can be computed by solving a linear system 
$VA=B$ where $V$ is the $n\times n$ (nonsingular) Vandermonde matrix defined by 
$V_{ij}=r_{j}^{i-1}$, $A$ and $B$ are the $n$-column vectors 
$A=(A_1,A_2,\ldots,A_n)$ and $B=(0,0,\ldots,1)$, that is:
\begin{equation}\label{vab}
\begin{bmatrix}1 & 1 & \cdots & 1\\r_1 & r_2 & \cdots & r_n\\r_1^2 & r_2^2 & \cdots & r_n^2 \\ \vdots & \vdots & \vdots & \vdots \\
r_1^{n-1} & r_2^{n-1} & \cdots & r_n^{n-1}\end{bmatrix}
\begin{bmatrix}A_1\\A_2\\A_3\\ \vdots \\ A_n\end{bmatrix} = 
\begin{bmatrix}0\\0\\0\\ \vdots \\ 1\end{bmatrix}
\end{equation}
So, vector $A$ is the last ($n$-th) column of the inverse of $V$.
\end{theorem}

\begin{proof} We use induction on $n$ to prove (\ref{convsum}), which is valid for $n=2$, as shown in (\ref{r1r2neqb}). 
Suppose (\ref{convsum}) is valid for $n=k$, and we prove it for $n=k+1$:
\begin{eqnarray*}
h_1*h_2*\cdots*h_k*h_{k+1} & = & (h_1*h_2*\cdots*h_k)*h_{k+1} \\
                                              & = & (A_1h_1+A_2h_2+\cdots+A_kh_k)*h_{k+1}\\
                                              & = & A_1(h_1*h_{k+1})+A_2(h_2*h_{k+1}) + \cdots + A_k(h_k*h_{k+1})\\
                                              & = & A_1(B_1h_1+C_1h_{k+1}) + A_2(B_2h_2+C_2h_{k+1}) + \cdots + A_k(B_kh_k+C_kh_{k+1})\\
                                              & = & (A_1B_1)h_1 + (A_2B_2)h_2 + \cdots + (A_kB_k)h_k + (A_1C_1+\cdots+A_kC_k)h_{k+1},
\end{eqnarray*}
and then (\ref{convsum}) is proved. To prove (\ref{vab}) we take the $i$-th derivative at $t=0^+$ on both sides of (\ref{convsum}) 
to get: 
\[
(h_1*h_2*\cdots*h_n)^{(i)}(0^+) = A_1h_1^{(i)}(0^+) +A_2h_2^{(i)}(0^+) + \cdots + A_nh_n^{(i)}(0^+), \quad i=0,1,2,\ldots, n-1.
\]
Applying Theorem~\ref{conv0+} to left side of equation above and using the fact that $h_j^{(i)}(0^+)=r_j^{i}$ we get (\ref{vab}).

\end{proof}

Now we consider the more general convolution $h_1*h_2*\cdots*h_n$, $n\ge2$, where there is the possibility 
of some $h_i$ to be repeated in the convolution, that is $h_i=h_j$ for some $i\neq j$. 
We initially consider some facts about the so-called  ``convolution power" 
(or ``$n$-fold" convolution \cite{chung,wikifold}) of exponentials, that is, the convolution of $h$, 
as defined in (\ref{hexp}), repeated between itself $n$ times, and we represent it by $h^{*n}$
(in Equation~(\ref{r1r2eq}) we have a formula for $h^{*2}$).

\begin{lemma}\label{foldconv}\em The convolution power of $n$ exponentials $h(t)=e^{rt}\sigma(t)$, 
denoted by $h^{*n}$, is given by
\[
h^{*n}(t)=\underbrace{(h*h*\cdots*h)}_{n \text{ terms}}(t) = \frac{1}{(n-1)!}t^{n-1}h(t), \quad n\ge 1
\]
\end{lemma}
\begin{proof} By induction on $n$. It is trivially true for $n=1$ and suppose it is valid for $n=k$, then: 
\begin{eqnarray*}
h^{*(k+1)}(t) = (h^{*k}*h)(t) & = & \int_0^{t}\frac{1}{(k-1)!}\tau^{k-1}e^{r\tau}e^{r(t-\tau)}d\tau, \quad t>0 \\
                                    & = & \frac{e^{rt}}{(k-1)!}\int_{0}^{t}\tau^{k-1}d\tau \\
                                    & = & \frac{1}{k(k-1)!}t^ke^{rt} = \frac{1}{k!}t^kh(t).
\end{eqnarray*}
\end{proof}

The Lemma bellow shows a generalization of Theorem~\ref{conv0+} applied to the convolution 
power of $h$:
\begin{lemma}\em\label{foldconvdot} Let be $h(t)=e^{rt}\sigma(t)$, then $i$-th derivative of $h^{*n}$, 
for $n\ge 2$, computed at $t=0^+$ and represented by $(h^{*n})^{(i)}(0^+)$, is given by:
\begin{equation}\label{hndot0+}
(h^{*n})^{(i)}(0^+) = \begin{cases} 0, & i=0,1, \ldots, n-2\\[0.2cm]
                                          \displaystyle\binom{i}{n-1}r^{i-n+1}, & i\ge n-1
                    \end{cases}
\end{equation}
\end{lemma}
\begin{proof} Equation~(\ref{hndot0+}) follows from Lemma~\ref{foldconv} by setting $k=n-1$ in the 
well-known formula:
\[
\frac{d^i}{dt^i}\left(\frac{t^k}{k!}e^{rt}\right)_{t=0} = \begin{cases}0, & i=0,1, \ldots, k-1\\[0.2cm]
                                          \displaystyle\binom{i}{k}r^{i-k}, & i\ge k
                    \end{cases}
\]
\end{proof}

Now we analyse how it would be like the convolution $h_1^{*n_1}*h_2^{*n_2}$, where 
$h_1(t)=e^{r_1t}\sigma(t)$ and $h_2(t)=e^{r_2t}\sigma(t)$, with $r_1\neq r_2$, 
that is the convolution between the ``$n_1$-power" convolution of $h_1$ with 
the ``$n_2$-power" convolution of $h_2$ when $h_1\neq h_2$.

\begin{lemma}\em\label{h1n1h2n2} Let be $h_1(t)=e^{r_1t}\sigma(t)$ and $h_2(t)=e^{r_2t}\sigma(t)$, 
with $r_1\neq r_2$, the convolution between the $n_1$-power convolution of 
$h_1$ and the $n_2$-power convolution of $h_2$, denoted by $h_1^{*n_1}*h_2^{*n_2}$, is given by:
\begin{eqnarray*}
h_1^{*n_1}*h_2^{*n_2} & = & \underbrace{(h_1*h_1*\cdots*h_1)}_{n_1\text{ terms}}*
                                             \underbrace{(h_2*h_2*\cdots*h_2)}_{n_2\text{ terms}} \nonumber\\
                                  & = & (A_1h_1+A_2h_1^{*2} + \cdots A_{n_1}h_1^{*n_1}) + 
                                            (B_1h_2 + B_2h_2^{*2} + \cdots B_{n_2}h_1^{*n_2})
\end{eqnarray*}
\end{lemma}
\begin{proof} We prove by induction on $(n_1,n_2)$. It is true for $(n_1,n_2)=(1,1)$ as shown in (\ref{r1r2neqb}).
\begin{enumerate}
\item Induction on $n_1$: Valid for $n_1=k$ and $n_2=1$. Let it be $n_1=k+1$:
\begin{eqnarray*}
h_1^{*(k+1)}*h_2 = h_1*(h_1^{*k}*h_2) & = & h_1*(A_1h_1+A_2h_1^{*2} + \cdots A_{k}h_1^{*k} + B_1h_2)\\
                        & = & A_1h_1^{*2}+A_2h_1^{*3} + \cdots A_{k+1}h_1^{*(k+1)} + B_1(h_1*h_2) \\
                        & = & A_1h_1^{*2}+A_2h_1^{*3} + \cdots A_{k+1}h_1^{*(k+1)} + B_1(C_1h_1+C_2h_2)\\
                       & = & (B_1C_1)h_1 + A_1h_1^{*2}+A_2h_1^{*3} + \cdots A_{k+1}h_1^{*(k+1)} + (B_1C_2)h_2
\end{eqnarray*}

\item Induction on $n_2$: Valid for generic $n_1$ and $n_2=k$. Let it be $n_2=k+1$: 
Since $h_1^{*n_1}*h_2^{*(k+1)} = (h_1^{*n_1}*h_2^{*k})*h_2$, then:
\begin{eqnarray*}
(h_1^{*n_1}*h_2^{*k})*h_2 & = & [(A_1h_1+A_2h_1^{*2} + \cdots + A_{n_1}h_1^{*n_1}) + 
                                                     (B_1h_2 + B_2h_2^{*2} + \cdots + B_{k}h_2^{*k})]*h_2\\
                                      & = & \underbrace{A_1(h_1*h_2)+A_2(h_1^{*2}*h_2) + \cdots + 
                                       A_{n_1}(h_1^{*n_1}*h_2)}_{\text{Rearranged as }
                                       (C_1h_1+ C_2h_1^{*2}+ \cdots + C_{n_1}h_1^{*n_1} + Dh_2)}
                                       + B_1h_2^{*2} + B_2h_2^{*3} + \cdots + B_{k}h_2^{*(k+1)}\\
                                       & = & (C_1h_1+C_2h_1^{*2}+\cdots+C_{n_1}h_1^{*n_1}) + 
                                       (Dh_2 + B_1h_2^{*2} + B_2h_2^{*3} + \cdots + B_{k}h_2^{*(k+1)})
\end{eqnarray*}
\end{enumerate}
\end{proof}

We now prove the general result about the power convolution of $n$ exponential signals as show in 
(\ref{hexp}) which is a generalization of Theorem~\ref{vandermn}:

\begin{theorem}\em\label{confvandn}
The convolution between $n\ge 2$ exponentials signals $\{h_1, h_2, \ldots, h_n\}$, with 
$h_i(t)=e^{r_it}\sigma(t)$, $r_i\in\MC$ and 
$q$ distinct $h_s$, each of them repeated $n_s$ times, so that $n_1+n_2+\cdots n_q=n$, is given by
\begin{equation} \label{convsumr} 
h_1^{*n_1}*h_2^{*n_2}*\cdots*h_q^{*n_q} = \sum_{j=1}^{n_1}A_{1j}h_1^{*j}+\sum_{j=1}^{n_2}A_{2j}h_2^{*j} + 
                                                                  \cdots + \sum_{j=1}^{n_q}A_{qj}h_q^{*j},                                                 
\end{equation}
where $A_{sj}\in\MC$ are scalars that can be computed by solving a linear system $VA=B$ where $V$ is the 
$n\times n$ nonsingular confluent (or generalized) Vandermonde matrix defined by 
$V=\begin{bmatrix}V_1 & V_2 & \cdots & V_q\end{bmatrix}$, where each block $V_s$ is the $n\times n_s$ matrix 
whose entries are defined by
\[
(V_s)_{ij} = \begin{cases}0, & i< j\\[0.2cm]
\displaystyle\binom{i-1}{j-1}r_s^{i-j}, & i\ge j
                                                     \end{cases}
\]
$A$ and $B$ are the $n$-column vectors 
$A=(A_1,A_2,\ldots,A_q)$, each $A_s$ is a $n_s$-column vector, and $B=(0_1,0_2,\ldots,B_q)$, 
where $0_s$ are $n_s$-column zero vectors and $B_q$ is the $n_q$-column vector 
$(0,0,\cdots,1)$ that is:
\begin{equation}\label{cvab}
\begin{bmatrix}V_1 & V_2 & \cdots & V_q\end{bmatrix}
\begin{bmatrix}A_1\\A_2\\A_3\\ \vdots \\ A_q\end{bmatrix} = 
\begin{bmatrix}0_1\\0_2\\0_3\\ \vdots \\ B_q\end{bmatrix}
\end{equation}
So, vector $A$ is the last ($n$-th) column of the inverse of $V$. Alternatively, using 
Lemma~\ref{foldconv}, we can rewrite 
(\ref{convsumr}) as 
\begin{equation}\label{convsump}
h_1^{*n_1}*h_2^{*n_2}*\cdots*h_q^{*n_q}= p_1h_1 + p_2h_2 + \cdots + p_qh_q
\end{equation}
where each $p_s$, $s=1,\ldots, q$, is a polynomial defined as
\[
p_s(t) = \sum_{j=1}^{n_s}A_{sj}\frac{t^{j-1}}{(j-1)!}
\]
\end{theorem}

\begin{proof} We use induction on $q$ to prove (\ref{convsumr}), which is valid for $q=2$, 
as shown in Lemma~\ref{h1n1h2n2}. 
Suppose (\ref{convsumr}) is valid for $q=k$, and we prove it for $q=k+1$:
\begin{eqnarray*}
h_1^{*n_1}*h_2^{*n_2}*\cdots*h_k^{*n_k}*h_{k+1}^{*n_{k+1}} & = & (h_1^{*n_1}*h_2^{*n_2}*\cdots*
                               h_k^{*n_k})*h_{k+1}^{*n_{k+1}} \\
                                & = & \left(\sum_{j=1}^{n_1}A_{1j}h_1^{*j}+\sum_{j=1}^{n_2}A_{2j}h_2^{*j} + \cdots +
                                          \sum_{j=1}^{n_k}A_{kj}h_k^{*j}\right)*h_{k+1}^{*n_{k+1}}\\
                                & = & \sum_{j=1}^{n_1}A_{1j}(h_1^{*j}*h_{k+1}^{*n_{k+1}})+\sum_{j=1}^{n_2}A_{2j}(h_2^{*j}*
                                          h_{k+1}^{*n_{k+1}})+ \cdots + \sum_{j=1}^{n_k}A_{kj}(h_k^{*j}*h_{k+1}^{*n_{k+1}})\\
                                & = & \sum_{j=1}^{n_1}B_{1j}h_1^{*j}+\sum_{j=1}^{n_2}B_{2j}h_2^{*j} + \cdots + 
                                          \sum_{j=1}^{n_k}B_{kj}h_{k}^{*j} + \sum_{j=1}^{n_{k+1}}B_{(k+1)j}h_{k+1}^{*j}
\end{eqnarray*}
and the (\ref{convsumr}) is proved. To prove (\ref{cvab}) we take the $i$-th derivative at $t=0^+$ 
on both sides of (\ref{convsumr}) to get: 
\begin{eqnarray*}
(h_1^{*n_1}*h_2^{*n_2}*\cdots*h_q^{*n_q})^{(i)}(0^+) & = & \sum_{j=1}^{n_1}A_{1j}(h_1^{*j})^{(i)}(0^+) +
\sum_{j=1}^{n_2}A_{2j}(h_2^{*j})^{(i)}(0^+) + \cdots  \\
& & + \sum_{j=1}^{n_q}A_{qj}(h_q^{*j})^{(i)}(0^+), \quad i=0,1,2,\ldots, n-1.
\end{eqnarray*}
Applying Theorem~\ref{conv0+} to left side of equation above and using the fact that 
$(h_k^{*1})^{(i)}(0^+)=h_k^{(i)}(0^+)=r_k^i$ 
along with Lemma~\ref{foldconvdot}, i.e., for $j\ge 2$:
\[
(h_s^{*j})^{(i)}(0^+)=\begin{cases} 0, & i=0,1, \ldots, j-2\\[0.2cm]
                                          \displaystyle\binom{i}{j-1}r_s^{i-j+1}, & i\ge j-1
                    \end{cases}
\]
we get (\ref{cvab}).
\end{proof}

\subsection{Solution of ordinary differential equations with constant coefficients}\label{edosec}

Consider the ordinary differential equation
\begin{equation}\label{odeeq}
y^{(n)} + a_{n-1}y^{(n-1)} + \cdots + a_1\dot y + a_0y = u, \quad a_i\in\MR
\end{equation}
which models an $n$ order (causal) linear time invariant (LIT) system with input signal $u$ 
and output signal $y$. The impulse response ($h$) for this system is given by the convolution \cite{ss}:
\[
h = h_1*h_2*\cdots*h_n, \quad h_i(t)=e^{r_it}\sigma(t), \quad r_i\in\MC
\]
and $r_1, r_2, \ldots, r_n$ are the roots of the characteristic equation $x^n+a_{n-1}x^{n-1} + \cdots + a_1x+a_0=0$
associated to (\ref{odeeq}). Supposing that the characteristic equation has $q$ distinct roots $r_s$, each one
repeated $n_s$ times, so that $n_1+n_2+\cdots+n_q=n$, then we can obtain the impulse response 
$h$ by using Theorem~\ref{confvandn}, Equation~(\ref{convsump}), that is
\begin{equation}\label{impresp}
h = p_1h_1+p_2h_2+\cdots+p_qh_q, \quad h_s(t)=e^{r_st}, \quad p_s(t)=\sum_{j=1}^{n_s}A_{sj}\frac{t^{j-1}}{(j-1)!},
\quad t>0
\end{equation}
where $A_{sj}$, $j=1,\ldots,n_s$ and $s=1,\ldots,q$ are calculated by solving the Vandermonde system (\ref{cvab}).

The complete solution of (\ref{odeeq}) is generally written as 
\begin{equation}\label{odesol}
y = y_h + y_p
\end{equation}
where $y_h$ is the homogeneous (or zero input) solution and $y_p$ is  a particular solution, i.e., it 
depends on input signal $u$.
When solving (\ref{odeeq}) for $t\ge0$, the particular solution $y_p$ can be written as 
\begin{equation}\label{ypsol}
y_p(t) = \int_{0}^tu(\tau)h(t-\tau)d\tau = [(u\sigma)*h](t), \quad\text{where } (u\sigma)(t)=
\begin{cases}0, & t < 0\\u(t), & t> 0\end{cases}
\end{equation}
The homogeneous solution ($y_h$) has the same format of (\ref{impresp}), that is
\begin{equation}\label{yhsol}
y_h = \bar p_1h_1 + \bar p_2h_2 + \cdots + \bar p_qh_q, \quad 
h_s(t)=e^{r_st}, \text{ and } \bar p_s(t)=\sum_{j=1}^{n_s}\bar A_{sj}\frac{t^{j-1}}{(j-1)!}.
\end{equation}
Therefore to solve (\ref{odeeq}) we need to obtain $y_h$, which is equivalent to obtain the constants 
$\bar A_{sj}$ in (\ref{yhsol}), and then compute $y_p$, by evaluating the convolution 
``$(u\sigma)*h$" as showed in (\ref{ypsol}). To find $y_h$ we use the fact that the particular solution 
$y_p$ is a convolution between $n+1$ signals, 
namely, ``$(u\sigma)*h_1*h_2*\cdots*h_n$" , and conclude, by using Theorem~\ref{conv0+}, that:
\[
y_p(0^+)=\dot y_p(0^+)=\ddot y_p(0^+)=\cdots=y_p^{(n-1)}(0^+)=0
\]
and so, using these conditions in (\ref{odesol}), we get:
\[
y(0^+) = y_h(0^+), \quad \dot y(0^+)=\dot y_h(0^+), \quad \ddot y(0^+)=\ddot y_h(0^+), \quad\cdots\quad 
y^{(n-1)}(0^+)=y_h^{(n-1)}(0^+).
\]
This set of conditions on $y_h$ can be used  to find the constants $\bar A_{sj}$ in (\ref{yhsol}) since the 
``initial values" $y(0), \dot y(0), \ddot y(0), \ldots, y^{(n-1)}(0)$ are generally known when solving
(\ref{odeeq}) for $t\ge0$. This implies that the constants $\bar A_{sj}$, $s=1,\ldots q$ and $j=1,\ldots,n_s$, 
can be computed by solving a Vandermonde system 
like the one showed in Theorem~\ref{confvandn}, that is $V\bar A=\bar B$, where the Vandermonde
matrix $V$ is the same one used to compute the impulse response $h$, $\bar A$ is the $n\times 1$
vector composed by the $\bar A_{sj}$'s and the vector $\bar B$, differently from the one used to compute $h$, it
is now defined as $\bar B=(y(0),\dot y(0), \ddot y(0), \cdots, y^{n-1}(0))$.

Finally, in order to obtain the complete solution $y$ for (\ref{odeeq}) as shown in (\ref{odesol}), 
we need to compute the 
particular solution ``$y_p=(u\sigma)*h$", that is the convolution between the input signal $u\sigma$ and the 
impulse response $h$, and to avoid solving a convolution integral we can use the result of 
Theorem~\ref{confvandn}, by writing, if possible, the signal 
``$u\sigma$" as a convolution (or a finite sum) of exponential signals of type ``$e^{rt}\sigma(t)$", 
for some $r\in\MC$.
In this situation, as shown in examples in Section~\ref{examplea} bellow, we increase the order of the 
Vandermonde matrix, as defined in Theorem~\ref{confvandn}, depending on how many 
``exponential modes" exists in the input signal ``$u\sigma$".

\section{Convolution between discrete time exponential signals}\label{discrconv}
In the context of discrete time signals we consider the exponential signal $e:\MZ\to \MC$ defined as
\begin{equation}\label{eexp}
e(k) = r^k\sigma(k), \quad r\neq 0\in\MC,\quad \sigma(k)=\begin{cases}0, & k<0\\1, & k\ge 0\end{cases}
\end{equation}
And also consider the signal defined as a right shift of ``$e$" by one unit, that is $h=[e]_1$, or:
\begin{equation}\label{dhexp}
h(k) = r^{k-1}\sigma(k-1),
\end{equation}
which is well known to appear as the impulse response of (causal) linear time invariant systems  (LTI) 
modeled by a first order difference equation, since it satisfies the relationship $h(k+1)=rh(k)+\delta(k)$. 
Now lets consider the convolution between two signals of 
this kind, that is, let be $h_1(k)=r_1^{k-1}\sigma(k-1)$ and $h_2(k)=r_2^{k-1}\sigma(k-1)$, 
with $r_1\neq 0$ and $r_2\neq 0$.
Since both of them are time shift of exponentials as defined in (\ref{eexp}), we can write
 $h_1=[e_1]_1$ and $h_2=[e_2]_1$, where $e_1(k)=r_1^k\sigma(k)$ and $e_2=r_2^k\sigma(k)$, and then:
\[
h_1*h_2=[e_1]_1*[e_2]_1 = (e_1*[\delta]_1)*(e_2*[\delta]_1) = (e_2*e_2)*([\delta]_1*[\delta]_1)=
(e_1*e_2)*[\delta]_2 = [e_1*e_2]_2
\]
therefore, $h_1*h_2$ can be obtained by a right time shift of $e_1*e_2$ by two units. We develop 
$e_1*e_2$ instead, noting  that $(e_1*e_2)(k)=0$ for $k<0$, since both $e_1(k)$ and $e_2(k)$ are 
null for $k<0$ and 
\begin{equation}\label{e1*e2}
(e_1*e_2)(k) = \sum_{j=0}^kr_1^jr_2^{k-j}, \quad\text{for } k\ge 0
\end{equation}
Additionally we also have that $(e_1*e_2)(0)=r_1^0r_2^0=1$. Then, before solving this summation, 
we note that the convolution $h_1*h_2$ is 
such that $(h_1*h_2)(k)=0$ for $k\le 0$, and, more importantly:
\begin{eqnarray}
(h_1*h_2)(1)& = & 0\label{dh1h21}\\
(h_1*h_2)(2)& = & 1\label{dh1h22}
\end{eqnarray}
since $h_1*h_2$ is a right shift of $e_1*e_2$ by two units.

We now develop the summation in (\ref{e1*e2}) by considering two cases:
\begin{enumerate}
\item $r_1\neq r_2$ (or $e_1\neq e_2$):
\begin{eqnarray*}
(e_1*e_2)(k) & = & r_2^k[1+(r_1/r_2) + (r_1/r_2)^2 + \cdots + (r_1/r_2)^k]\\
                     & = & r_2^k\frac{(r_1^{k+1}/r_2^{k+1})-1}{(r_1/r_2)-1}\\
                     & = & \frac{r_1^{k+1}-r_2^{k+1}}{r_1-r_2}
\end{eqnarray*}
and since $h_1*h_2=[e_1*e_2]_{2}$, then $(h_1*h_2)(k) = (e_1*e_2)(k-2)$ or:
\begin{eqnarray}
(h_1*h_2)(k) & = & \frac{1}{r_1-r_2}r_1^{k-1} \sigma(k-1)+ \frac{1}{r_2-r_1}r_2^{k-1}\sigma(k-1) \label{dr1r2neq}, 
\text{ or}\\
(h_1*h_2)(k) & = & A_1h_1(k) + A_2h_2(k), \quad A_1=\frac{1}{r_1-r_2}\text{ and } A_2=\frac{1}{r_2-r_1} \label{dr1r2neqb}
\end{eqnarray}
\begin{remark}\em
Note that in case where $r_1$ and $r_2$ is a complex conjugate pair, represented by 
$\alpha\pm j\omega=Re^{\pm j\phi}$, we get from (\ref{dr1r2neq}) that 
$(h_1*h_2)(k) = (R^{k-1}/\omega)\sin[(k-1)\phi]$, for $k\ge1$.
\end{remark}

From Equation~(\ref{dr1r2neqb}) we see that, in case that $r_1\neq r_2$, the convolution 
$h_1*h_2$ can be written as 
a {\em linear combination} of signals $h_1$ and $h_2$, and this fact, along with conditions (\ref{dh1h21}) and 
(\ref{dh1h22}), can be used to find the scalars $A_1$ and $A_2$, without the need of solving the convolution sum 
(\ref{e1*e2}), as shown bellow:
\begin{eqnarray*}
(h_1*h_2)(1) & = & A_1h_1(1) + A_2h_2(1)= A_1 + A_2 = 0\\
(h_1*h_2)(2) & = & A_1h_1(2) + A_2h_2(2) = A_1r_1 + A_2r_2 = 1
\end{eqnarray*}
And then:
\begin{equation}\label{dvanderm2}
\begin{bmatrix} 1 & 1\\r_1 & r_2\end{bmatrix}\begin{bmatrix}A_1 \\ A_2\end{bmatrix} = \begin{bmatrix}0 \\ 1\end{bmatrix}
\implies
\begin{bmatrix}A_1 \\ A_2\end{bmatrix} = \begin{bmatrix} 1 & 1\\r_1 & r_2\end{bmatrix}^{-1}\begin{bmatrix}0 \\ 1 \end{bmatrix}.
\end{equation}
Solving (\ref{dvanderm2}) we get $A_1$ and $A_2$ as shown in (\ref{dr1r2neqb}).
\item $r_1=r_2=r$ (or $e_1=e_2=e$): 
\[
(e*e)(k) = r^k\sum_{j=0}^kr^j*r^{-j} = (k+1)r^{k}, \quad k\ge 0
\]
and then, since $h_1=h_2=h = [e]_{1}$, $(h*h)(k)=(e*e)(k-2)$ is given by
\begin{equation}\label{dr1r2eq}
(h*h)(k) = \begin{cases}0, & k\le 1\\(k-1)r^{k-2}, & k\ge 2\end{cases}
\end{equation}

\end{enumerate}

Now we consider a generalization of the results above for a convolution of $n\ge2$ exponential signals as 
shown in (\ref{dhexp}).
We start by finding a generalization for conditions (\ref{dh1h21}) and (\ref{dh1h22}) applied to the convolution 
$h_1*h_2*\cdots*h_n$, with $h_i(k) = r_i^{k-1}\sigma(k-1)$ and $n\ge2$:

\begin{theorem}\label{dconv1+}\em
Consider the convolution $h_1*h_2*\cdots*h_n$, $n\ge 2$ and each $h_i(k)=r_i^{k-1}\sigma(k-1)$, $r_i\neq 0\in\MC$. 
Then we have
\[
(h_1*h_2*\cdots*h_n)(k) = \begin{cases}0, & k\le n-1\\
                                                                    1, & k=n
                                                      \end{cases}
\]
\end{theorem}
\begin{proof} Defining $e_i(k)= r_i^k\sigma(k)$, we note that $h_i=[e_i]_1$ and then
\[
(h_1*h_2*\cdots*h_n) = ([e_1]_1*[e_2]_1*\cdots*[e_n]_1) = [e_1*e_2*\cdots*e_n]_n
\]
that is, $h_1*h_2*\cdots*h_n$ is a time shift right of $e_1*e_2*\cdots*e_n$ by $n$ units, and since 
$(e_1*e_2*\cdots*e_n)(k)=0$ for $k<0$ and $(e_1*e_2*\cdots*e_n)(0)=1$ the result is proved.
\end{proof}

In the following we will find a formula for computing the convolution $h_1*h_2*\cdots*h_n$ for $n\ge 2$ and 
$h_j(k)=r_j^{k-1}\sigma(k-1)$ with $r_j\in\MC$. To begin with, we consider the case where 
$h_i\neq h_j$ for $i\neq j$, which implies $r_i\neq r_j$ for $i\neq j$, and it is just a generalization 
of Equation~(\ref{dvanderm2}):

\begin{theorem}\em\label{dvandermn}
The convolution between $n\ge 2$ exponentials signals $h_j(k)=r_j^{k-1}\sigma(k-1)$, $j=1,2,\ldots,n$, with 
$r_j\neq0\in\MC$ 
and $h_i\neq h_j$ for $i\neq j$, is given by
\begin{equation}\label{dconvsum}
h_1*h_2*\cdots*h_n=A_1h_1+A_2h_2 + \cdots + A_nh_n,
\end{equation}
where $A_j\in\MC$ are scalars that can be computed by solving a linear system $VA=B$ where 
$V$ is the $n\times n$ 
(nonsingular) Vandermonde matrix defined by $V_{ij}=r_{j}^{i-1}$, $A$ and $B$ are the $n$-column vectors 
$A=(A_1,A_2,\ldots,A_n)$ and $B=(0,0,\ldots,1)$, 
that is:
\begin{equation}\label{dvab}
\begin{bmatrix}1 & 1 & \cdots & 1\\r_1 & r_2 & \cdots & r_n\\r_1^2 & r_2^2 & \cdots & r_n^2 \\ \vdots & \vdots & \vdots & \vdots \\
r_1^{n-1} & r_2^{n-1} & \cdots & r_n^{n-1}\end{bmatrix}
\begin{bmatrix}A_1\\A_2\\A_3\\ \vdots \\ A_n\end{bmatrix} = 
\begin{bmatrix}0\\0\\0\\ \vdots \\ 1\end{bmatrix}
\end{equation}
So, vector $A$ is the last ($n$-th) column of the inverse of $V$.
\end{theorem}

\begin{proof} We use induction on $n$ to prove (\ref{dconvsum}), which is valid for $n=2$, as shown in (\ref{dr1r2neqb}). 
Suppose (\ref{dconvsum}) is valid for $n=k$, and we prove it for $n=k+1$ following the same reasoning we used to prove
(\ref{convsum}) in Theorem~\ref{vandermn}. To prove (\ref{dvab}) we apply the result of Theorem~\ref{dconv1+} to 
Equation~(\ref{dconvsum}). Taking the value at $k=i$ on both sides of (\ref{dconvsum}) we have:
\[
(h_1*h_2*\cdots*h_n)(i) = A_1h_1(i) +A_2h_2(i) + \cdots + A_nh_n(i), \quad i=1,2,\ldots, n.
\]
Using Theorem~\ref{dconv1+} and the fact that $h_j(i)=r_j^{i-1}$ we get (\ref{dvab}).

\end{proof}

Now we consider the more general convolution $h_1*h_2*\cdots*h_n$, $n\ge2$, 
where there is the possibility of some $h_i$ to be 
repeated in the convolution, that is $h_i=h_j$ for some $i\neq j$. To begin with, we consider some facts about 
``$n$-power" convolution of discrete time exponentials, that is, the convolution of $h$, as defined in (\ref{dhexp}), 
repeated between itself $n$ times, that we represent it by $h^{*n}$ (in Equation~(\ref{dr1r2eq}) we have a 
formula for $h^{*2}$).
The Lemma bellow shows a generalization of Theorem~\ref{dconv1+} applied to the ``$n$-power" convolution of the 
exponential signal:
\begin{lemma}\label{dfoldconv}\em The power convolution of $n\ge1$ exponentials $e(k)=r^{k}\sigma(k)$, 
$r\neq0\in\MC$, denoted by $e^{*n}$, is given by
\[
e^{*n}(k)=\underbrace{(e*e*\cdots*e)}_{n \text{ terms}}(k) = \begin{cases}0, & k<0\\[0.2cm]
                                                                                            \displaystyle\binom{n-1+k}{n-1}r^{k}, & k\ge 0
                                                                                           \end{cases}
\]
or, in a more compact notation 
\begin{equation}\label{nfolde}
e^{*n}(k) = \binom{n-1+k}{n-1}r^k\sigma(k) = \binom{n-1+k}{n-1}e(k)
\end{equation}
\end{lemma}
\begin{proof}By induction on $n$. It is trivially true for $n=1$ and suppose it is valid for $n=p$ then
\[
e^{*p}(k) = \binom{p-1+k}{p-1}r^{k}\sigma(k)
\]
Obviously $e^{*(p+1)}(k)=0$ for $k<0$ since $e(k)=0$ for $k<0$; for $k\ge 0$ we have:
\begin{eqnarray*}
e^{*(p+1)}(k) = (e^{*p}*e)(k) & = & \sum_{j=0}^{k}e^{*p}(j)e(k-j) \\
                                    & = & \sum_{j=0}^{k}\binom{p-1+j}{p-1}r^{j}r^{(k-j)}\\
                                    & = & r^{k}\sum_{j=0}^{k}\binom{p-1+j}{p-1}\\
                                    & = & r^{k}\binom{p+k}{p}
\end{eqnarray*}
In the last step of the proof above we used the following well-known fact about sum of binomial 
coefficients \cite{chung}: 
\[
\sum_{j=0}^{k}\binom{p-1+j}{p-1} = \displaystyle\binom{p-1}{p-1}+\binom{p}{p-1}+\binom{p+1}{p-1}+\cdots+
\binom{p-1+k}{p-1}=\binom{p+k}{p}
\]

\end{proof}

\begin{corollary}\em\label{dfoldconvm}

If we consider the $n$-power convolution of exponentials $h(k)=r^{k-1}\sigma(k-1)$, that is $h=[e]_1$, we have:
\begin{equation}\label{nfoldh}
h^{*n}(k)= \begin{cases}0, & k\le n-1\\[0.2cm]
                               \displaystyle\binom{k-1}{n-1}r^{k-n}, & k\ge n
             \end{cases}
\end{equation}
equivalently
\begin{equation}\label{nfoldh1}
h^{*n}(k) = \frac{1}{r^{n-1}}\binom{k-1}{n-1}h(k), \quad n\ge 1
\end{equation}
since it is assumed that $\binom{k-1}{n-1}=0$ for $k=1,2,\ldots, n-1$.
\begin{proof} Since $h=[e]_1$, then 
\begin{eqnarray*}
h^{*n} & = & \underbrace{(h*h*\cdots*h)}_{n\text{ terms}}\\
       & = & [e]_1*[e]_1*\cdots*[e]_1\\
       & = & [e*e*\cdots*e]_n\\
       & = & [e^n]_{n}
\end{eqnarray*}
that is, $h^{*n}$ is $e^{*n}$ (right) shifted $n$ units. Then we have by setting $k:=k-n$ in (\ref{nfolde}):
\begin{equation}\label{nfoldh2}
h^{*n}(k) = \binom{k-1}{n-1}r^{k-n}\sigma(k-n)
\end{equation}
which is equivalent to (\ref{nfoldh}). To obtain (\ref{nfoldh1}), we note that $\binom{k-1}{n-1} = 0$, 
for $k=1,2,\ldots n-1$, and so (\ref{nfoldh2}) can be rewriten as 
\[
h^{*n}(k) = \binom{k-1}{n-1}r^{k-n}\sigma(k-1) = \frac{1}{r^{n-1}}\binom{k-1}{n-1}r^{k-1}\sigma(k-1) = 
\frac{1}{r^{n-1}}\binom{k-1}{n-1}h(k)
\]
\end{proof}

\end{corollary}

Now we analyse how it would be like the convolution $h_1^{*n_1}*h_2^{*n_2}$, where 
$h_1(k)=r_1^{k-1}\sigma(k-1)$ and 
$h_2(k)=r_2^{k-1}\sigma(k-1)$, with $r_1\neq r_2$: 

\begin{lemma}\em\label{dh1n1h2n2} Let be $h_1(t)=r_1^{k-1}\sigma(k-1)$ and $h_2(k)=r_2^{k-1}\sigma(k-1)$, 
with $r_1\neq r_2$, the convolution between the $n_1$-power convolution of $h_1$ and the 
$n_2$-power convolution of $h_2$, 
denoted by $h_1^{*n_1}*h_2^{*n_2}$, is given by:
\begin{eqnarray*}
h_1^{*n_1}*h_2^{*n_2} & = & \underbrace{(h_1*h_1*\cdots*h_1)}_{n_1\text{ terms}}*
                                             \underbrace{(h_2*h_2*\cdots*h_2)}_{n_2\text{ terms}} \nonumber\\
                                  & = & (A_1h_1+A_2h_1^{*2} + \cdots A_{n_1}h_1^{*n_1}) + 
                                  (B_1h_2 + B_2h_2^{*2} + \cdots B_{n_2}h_1^{*n_2})
\end{eqnarray*}
\end{lemma}
\begin{proof} We prove by induction on $(n_1,n_2)$. It is true for $(n_1,n_2)=(1,1)$ as shown in 
(\ref{dr1r2neqb}). The 
inductive step is the same one used in the proof of Lemma~\ref{h1n1h2n2} for the analog time case.
\end{proof}

In the following we prove the general result about the convolution of $n$ exponential signals as show in (\ref{dhexp}) 
which is a generalization of Theorem~\ref{dvandermn}:

\begin{theorem}\em\label{dconfvandn}
The convolution between $n\ge 2$ exponentials signals $h_i(k)=r_i^{k-1}\sigma(k-1)$, $i=1,2,\ldots,n$, with 
$r_i\neq0\in\MC$, and $q$ distinct $h_s$, each of them repeated $n_s$ times, so that 
$n_1+n_2+\cdots n_q=n$, is given by
\begin{equation} \label{dconvsumr} 
h_1^{*n_1}*h_2^{*n_2}*\cdots*h_q^{*n_q} = \sum_{j=1}^{n_1}A_{1j}h_1^{*j}+\sum_{j=1}^{n_2}A_{2j}h_2^{*j} + \cdots + 
                                                                 \sum_{j=1}^{n_q}A_{qj}h_q^{*j},                                                 
\end{equation}
where $A_{sj}\in\MC$ are scalars that can be computed by solving a linear system $VA=B$ where 
$V$ is the $n\times n$ 
(nonsingular) confluent (or generalized) Vandermonde matrix defined by 
$V=\begin{bmatrix}V_1 & V_2 & \cdots & V_q\end{bmatrix}$, 
where each block $V_s$ is the $n\times n_s$ matrix whose entries are defined by
\[
(V_s)_{ij} = \begin{cases}0, & i< j\\[0.2cm]
\displaystyle\binom{i-1}{j-1}r_s^{i-j}, & i\ge j
                                                     \end{cases}
\]
$A$ and $B$ are the $n$-column vectors 
$A=(A_1,A_2,\ldots,A_q)$, each $A_s$ is a $n_s$-column vector, and $B=(0_1,0_2,\ldots,B_q)$, 
where $0_s$ are 
$n_s$-column zero vectors and $B_q$ is the $n_q$-column vector $(0,0,\cdots,1)$
that is:
\begin{equation}\label{dcvab}
\begin{bmatrix}V_1 & V_2 & \cdots & V_q\end{bmatrix}
\begin{bmatrix}A_1\\A_2\\A_3\\ \vdots \\ A_q\end{bmatrix} = 
\begin{bmatrix}0_1\\0_2\\0_3\\ \vdots \\ B_q\end{bmatrix}
\end{equation}
So, vector $A$ is the last ($n$-th) column of the inverse of $V$. Alternatively, using 
Equation~(\ref{nfoldh1}), we can rewrite 
(\ref{dconvsumr}) as 
\begin{equation}\label{dconvsump}
h_1^{*n_1}*h_2^{*n_2}*\cdots*h_q^{*n_q} = p_1h_1 + p_2h_2 + \cdots + p_qh_q
\end{equation}
where each $p_s$, $s=1,\ldots, q$, is a polynomial defined as
\[
p_s(k) = \sum_{j=1}^{n_s}A_{sj}\frac{1}{r_s^{j-1}}\binom{k-1}{j-1}, \quad k\ge 1
\]
\end{theorem}

\begin{proof} We use induction on $q$ to prove (\ref{dconvsumr}), which is valid for $q=2$, as shown in Lemma~\ref{dh1n1h2n2}. 
The inductive step follows in the same way we did in the proof of Theorem~\ref{confvandn}. To prove (\ref{dcvab}) we evaluate Equation~(\ref{dconvsumr}) at $k=i$ to obtain:
\begin{eqnarray*}
(h_1^{*n_1}*h_2^{*n_2}*\cdots*h_q^{*n_q})(i) & = & \sum_{j=1}^{n_1}A_{1j}h_1^{*j}(i) + 
\sum_{j=1}^{n_2}A_{2j}h_2^{*j}(i) + \cdots  \\
& & + \sum_{j=1}^{n_q}A_{qj}h_q^{*j}(i), \quad i=1,2,\ldots, n.
\end{eqnarray*}
Applying the result of Theorem~\ref{dconv1+} to the left side of this equation  and using 
Lemma~\ref{dfoldconv}, Equation~(\ref{nfoldh}), that is for $j\ge 1$:
\[
h_s^{*j}(i)=\begin{cases} 0, & i\le j-1\\[0.2cm]
                                          \displaystyle\binom{i-1}{j-1}r_s^{i-j}, & i\ge j
                    \end{cases}
\]
we get (\ref{dcvab}).
\end{proof}

\subsection{Solution of difference equations with constant coefficients}

Consider the ``$n$ order" difference equation
\begin{equation}\label{deeq}
y(k+n) + a_{n-1}y(k+n-1) + \cdots + a_1y(k+1) + a_0y(k) = u(k)
\end{equation}
which models an $n$ order discrete time (causal) linear time invariant (LIT) system with input signal 
$u$ and output signal $y$.
The impulse response ($h$) for this system is given by the convolution \cite{ss}:
\[
h = h_1*h_2*\cdots*h_n, \quad h_i(k)=r_i^{k-1}\sigma(k-1), \quad r_i\neq0\in\MC
\]
and $r_1, r_2, \ldots, r_n$ are the roots of the characteristic equation $x^n+a_{n-1}x^{n-1} + \cdots + a_1x+a_0=0$
associated to (\ref{deeq}), which all are assumed to be non-zero.\footnote{Zero roots are discarded and order of 
the difference equation reduced by the amount of discarded roots. The final solution is then 
the solution of the reduced order equation right-shifted as many units as the number of zero roots of 
the characteristic equation (see examples in Section~\ref{exampled}).} 
Supposing that the characteristic equation has $q$ distinct non-zero roots $r_s$, each one
repeated $n_s$ times, so that $n_1+n_2+\cdots+n_q=n$, then we can obtain $h$ by using 
Theorem~\ref{dconfvandn},
Equation~(\ref{dconvsump}), that is
\begin{equation}\label{dimpresp}
h = p_1h_1+p_2h_2+\cdots+p_qh_p, \quad h_s(k)=r_s^{k-1}, \quad 
p_s(k) = \sum_{j=1}^{n_s}A_{sj}\frac{1}{r_s^{j-1}}\binom{k-1}{j-1}, \quad k\ge 1
\end{equation}
where $A_{sj}$, $j=1,\ldots,n_s$ and $s=1,\ldots,q$ are calculated by solving the Vandermonde system (\ref{dcvab}).

The solution of (\ref{deeq}) for $k\ge 0$ can be written as:
\begin{equation}\label{desol}
y = y_h + y_p
\end{equation}
where $y_h$ is the homogeneous (or zero input) solution and $y_p$ is a particular solution, i.e., it depends on the 
input signal $u$. When solving (\ref{desol}) for $k\ge0$, the particular solution can be written as
\begin{equation}\label{dypsol}
y_p(k) = \sum_{j=0}^k u(j)h(k-j) = [(u\sigma)*h](k), \quad\text{where }(u\sigma)(k)=\begin{cases}0, & k<0\\u(k), & k\ge0\end{cases}
\end{equation}
The homogeneous solution has the same format of (\ref{dimpresp}), that is
\begin{equation}\label{dyhsol}
y_h = \bar p_1\bar h_1 + \bar p_2\bar h_2 + \cdots + \bar p_q\bar h_q, \quad 
\bar h_s(k)=r_s^{k}, \text{ and } \bar p_s(k)=\sum_{j=0}^{n_s-1}\bar A_{sj}\frac{1}{r_s^j}\binom{k}{j}, \quad k\ge 0
\end{equation}
Therefore to solve (\ref{deeq}) we need to obtain $y_h$, which is equivalent obtain the constants $\bar A_{sj}$ in (\ref{dyhsol}),
and then obtain $y_p$, by evaluating the convolution ``$(u\sigma)*h$'' as shown in (\ref{dypsol}).
Since the particular solution $y_p$ is, in fact, a convolution between $n+1$ signals, 
namely, ``$(u\sigma)*h_1*h_2*\cdots*h_n$" , we conclude, by using Theorem~\ref{conv0+}, that:
\[
y_p(0)=y_p(1)=y_p(2)=\cdots=y_p(n-1)=0
\]
and so, by (\ref{desol}), we have that:
\[
y(0) = y_h(0), \quad y(1)=y_h(1), \quad y(2)=y_h(2), \quad\cdots\quad y(n-1)=y_h(n-1)
\]
which can be used in (\ref{dyhsol}) to find the constants $\bar A_{sj}$, $j=1,\ldots n_s$ and $s=1, \ldots q$,
since the ``initial values" $y(0), y(1), y(2), \ldots, y(n-1)$ are generally known when solving
(\ref{deeq}) for $k\ge 0$. In fact, constants $\bar A_{sj}$ are computed by solving a Vandermonde system 
like the one showed in Theorem~\ref{dconfvandn}, that is $V\bar A=\bar B$, where the Vandermonde
matrix $V$ is the same one used to compute the impulse response $h$, $\bar A$ is the $n\times 1$
vector composed by the $\bar A_{sj}$'s and the vector $\bar B$, differently from the one used to compute $h$, it
is defined as $\bar B=(y(0),y(1), y(2), \cdots, y(n-1))$.

Finally, in order to obtain the complete solution $y$ for (\ref{deeq}) as shown in (\ref{desol}), we need to compute the 
particular solution ``$y_p=(u\sigma)*h$", that is the convolution between the input signal $u\sigma$ and the 
inpulse response $h$, and this can be done by the result of Theorem~\ref{dconfvandn} if we can write the signal 
``$u\sigma$" as a convolution (or a sum) of exponential signals of type ``$r^k\sigma(k)$", for some $r\neq0\in\MC$.
In this situation, as shown in examples bellow, we increase the order of the Vandermonde matrix, as 
defined in Theorem~\ref{dconfvandn}, depending on how many ``exponential modes" exists in the input 
signal ``$u\sigma$". In Section~\ref{exampled} we apply these results to the resolution of some specific 
difference equations.

\section{Examples}\label{ddeqex}
Bellow we apply the results discussed in previous sections to the solution to some specific 
differential/difference equations.

\subsection{Differential Equations}\label{examplea}

\begin{example}\em 
Let be the second order initial value problem (IVP):
\begin{equation}\label{ex1}
\ddot y + 3\dot y + 2y = 1, \quad\text{with } y(0)=-1 \text{ and }\dot y(0)=2. 
\end{equation}
To find the solution $y$, we consider the characteristic equation is $x^2+3x+2=0$ whose roots as 
$r_1=-1$ and $r_2=-2$.
\begin{enumerate}
\item[(a)] Impulse response:
\(
h(t) = A_1e^{-t} + A_2e^{-2t},
\)
where $A_1$ and $A_2$ are computed as
\[
\begin{bmatrix}1 & 1 \\ -1 & -2\end{bmatrix}\begin{bmatrix}A_1\\ A_2\end{bmatrix} = \begin{bmatrix}0 \\ 1\end{bmatrix}
\implies \begin{bmatrix}A_1\\ A_2\end{bmatrix} = \begin{bmatrix}1 \\ -1\end{bmatrix}
\]
which implies $h(t) = e^{-t} - e^{-2t}$.

\item[(b)]Homogeneous solution:
\(
y_h(t) = B_1e^{-t} + B_2e^{-2t},
\)
where $B_1$ and $B_2$ are computed as:
\[
\begin{bmatrix}1 & 1 \\ -1 & -2\end{bmatrix}\begin{bmatrix}B_1\\ B_2\end{bmatrix} = \begin{bmatrix}y(0) \\ \dot y(0)\end{bmatrix}
=\begin{bmatrix}-1 \\ 2\end{bmatrix}\implies \begin{bmatrix}B_1\\ B_2\end{bmatrix} = \begin{bmatrix} 0 \\ -1\end{bmatrix}
\]
which implies $y_h(t) = - e^{-2t}$.

\item[(c)] Particular solution: $y_p=(u\sigma)*h$, and 
\(
(u\sigma)(t) = 1.\sigma(t)= e^{0t}\sigma(t), 
\)
then 
\[
y_p=(u\sigma)*h = h*(u\sigma) = h_1*h_2*h_3
\]
where $h_1(t)=e^{-t}\sigma(t), h_2(t)=e^{-2t}\sigma(t)$ and $h_3(t) = e^{0t}\sigma(t)$, or:
\[
y_p(t) = C_1e^{-t} + C_2e^{-2t} + C_3e^{0t}
\]
where $C_1, C_2$ and $C_3$ are compute as the solution of the ``augmented" Vandermonde system:
\[
\begin{bmatrix}1 & 1 & 1\\ -1 & -2 & 0\\ 1 & 4 & 0\end{bmatrix}
\begin{bmatrix}C_1\\ C_2 \\C_3\end{bmatrix} = \begin{bmatrix}0 \\ 0 \\ 1\end{bmatrix} \implies
\begin{bmatrix}C_1\\ C_2 \\C_3\end{bmatrix} = \begin{bmatrix}-1 \\ 0.5 \\ 0.5\end{bmatrix}
\]
which implies $y_p(t) = - e^{-t} + 0.5e^{-2t} + 0.5$.
\end{enumerate}

Finally, the solution for the IVP (\ref{ex1}) is $y=y_h+y_p$ or
\[
y(t) = -e^{-t} -0.5e^{-2t} + 0.5
\]
\end{example}

\begin{example}\em
Let be the following third order IVP
\begin{equation}\label{ex2}
\dddot y + 7\ddot y + 20\dot y + 24 y = \sin2t, \quad y(0)=0, \dot y(0)=1, \ddot y(0)=-3
\end{equation}
The characteristic equation is $x^3+7x+20x+24=0$ whose roots are $r_1=-3,r_2=-2+2i$ and $r_3=-2-2i$.
\begin{enumerate}
\item[(a)] Impulse response: $h(t) = A_1e^{-3t} + A_2e^{(-2+2i)t} + A_3e^{(-2-2i)t}$, and
\[
\begin{bmatrix}1 & 1 & 1\\ -3 & -2+2i & -2-2i\\ 9 & -8i & 8i\end{bmatrix}
\begin{bmatrix}A_1\\ A_2 \\A_3\end{bmatrix} = \begin{bmatrix}0 \\ 0 \\ 1\end{bmatrix} \implies 
\begin{bmatrix}A_1\\ A_2\\ A_3\end{bmatrix} = \begin{bmatrix}0.2\\ -0.1-0.05i \\ -0.1+0.05i\end{bmatrix}
\]
and then
\[
h(t) = 0.2e^{-3t} + (-0.1-0.05i)e^{(-2+2i)t} + (-0.1+0.05i)e^{(-2-2i)t}
\]
which (optionally) can be simplified to 
\[
h(t) = 0.2e^{-3t}-0.2e^{-2t}\cos2t + 0.1e^{-2t}\sin2t.
\]

\item[(b)] Homogeneous solution: $y_h(t) = B_1e^{-3t} + B_2e^{(-2+2i)t} + B_3e^{(-2-2i)t}$, and 
\[
\begin{bmatrix}1 & 1 & 1\\ -3 & -2+2i & -2-2i\\ 9 & -8i & 8i\end{bmatrix}
\begin{bmatrix}B_1\\ B_2 \\B_3\end{bmatrix} = \begin{bmatrix}0 \\ 1 \\ -3\end{bmatrix} \implies 
 \begin{bmatrix}B_1 \\ B_2 \\ B_3\end{bmatrix} =  \begin{bmatrix}0.2 \\ -0.1-0.3i \\ -0.1+0.3i\end{bmatrix}
\]
and then
\[
y_h(t) = 0.2e^{-3t} + (-0.1-0.3i)e^{(-2+2i)t} + (-0.1+0.3i)e^{(-2-2i)t}
\]
or 
\[
y_h(t) = 0.2e^{-3t} -0.2e^{-2t}\cos2t + 0.6e^{-2t}\sin2t.
\]

\item[(c)] Particular solution: $y_p = (u\sigma)*h$, since $u(t) = \sin2t$ we have two possibilites:
\[
u\sigma =2(h_4*h_5) \quad\text{or}\quad u\sigma = \frac{h_4+h_5}{2i}, \quad\text{where }
h_4(t)=e^{2it}\sigma(t) \text{ and } h_5(t)=e^{-2it}\sigma(t)
\]
using $u\sigma=2(h_4*h_5)$ we have 
\[
y_p = 2(h_1*h_2*h_3*h_4*h_5),
\]
where 
\[
h_1(t)=e^{-3t}\sigma(t), h_2(t) = e^{(-2+2i)t}\sigma(t), h_3(t) = e^{(-2-2i)t}\sigma(t), h_4(t)=e^{2it}\sigma(t), 
h_5(t)=e^{-2it}\sigma(t).
\]
So,  to compute $h_1*h_2*h_3*h_4*h_5$, we have the following (augmented) Vandemonde system:
\[
\begin{bmatrix}1 & 1 & 1 & 1 & 1\\ -3 & -2+2i & -2-2i & 2i & -2i\\ 9 & -8i & 8i & -4 & -4\\ 
-27 & 16+16i & 16-16i & -8i & 8i\\ 81 & -64 & -64 & 16 & 16\end{bmatrix}
\begin{bmatrix}C_1\\ C_2 \\C_3 \\ C_4 \\ C_5\end{bmatrix} = \begin{bmatrix}0 \\ 0 \\ 0 \\ 0 \\ 1\end{bmatrix} \implies
(2)\times\begin{bmatrix}C_1 \\ C_2 \\ C_3 \\ C_4 \\ C_5\end{bmatrix} = 
\begin{bmatrix} 0.0307692 \\ - 0.025i \\ 0.025i \\ - 0.0153846 + 0.0019231i \\ - 0.0153846 - 0.0019231i \end{bmatrix}
\]
Then the particular solution is\[
y_p(t) = 0.0307692e^{-3t} +0.05e^{-2t}\sin2t - 0.0307692\cos2t - 0.0038462\sin2t
\]
\end{enumerate}
Finally, the solution $y = y_h+y_p$ for the IVP (\ref{ex2}) is given by:
\[
y(t) = 0.2307692e^{-3t} - 0.2e^{-2t}\cos2t + 0.65e^{-2t}\sin2t - 0.0307692\cos2t - 0.0038462\sin2t
\]
\end{example}

\begin{example}\em 
Let be the following IVP
\begin{equation}\label{ex3}
\ddot y + 4y = t\cos2t, \quad y(0)=-2, \quad \dot y(0)=4
\end{equation}
whose characteristic equation is $x^2+4=0$ which implies $r_1=e^{2it}$ and $r_2=e^{-2it}$.

\begin{enumerate}
\item[(a)] Impulse response: $h(t) = A_1e^{2it} + A_2e^{-2it}$, and 
\[
\begin{bmatrix}1 & 1\\ 2i & -2i\end{bmatrix}\begin{bmatrix}A_1\\ A_2\end{bmatrix} = 
\begin{bmatrix}0 \\ 1\end{bmatrix} \implies
\begin{bmatrix}A_1\\ A_2\end{bmatrix} = \begin{bmatrix}-0.25i\\ 0.25i\end{bmatrix}
\]
Then 
\[
h(t) = -025ie^{2it} + 0.25ie^{-2it} = 0.5\sin2t.
\]

\item[(b)] Homogeneous solution: $y_h(t) = B_1e^{2it} + B_2e^{-2it}$, and 
\[
\begin{bmatrix}1 & 1\\ 2i & -2i\end{bmatrix}\begin{bmatrix}B_1\\ B_2\end{bmatrix} = 
\begin{bmatrix}-2 \\ 4\end{bmatrix} \implies
\begin{bmatrix}B_1\\ B_2\end{bmatrix} = \begin{bmatrix}-1-i\\ -1+i\end{bmatrix}
\]
Then
\[
y_h(t) = (-1-i)e^{2it} + (-1+i)e^{-2it} = -2\cos2t + 2\sin2t
\]

\item[(c)] Particular solution: $y_p=(u\sigma)*h$, and $u(t)=t\cos2t = t(e^{2it}+e^{-2it})/2$, or:
\[
u = 0.5(u_1+u_2), \quad u_1(t) = te^{2it}, \quad u_2(t)=te^{-2it}
\]
and so, $y_p = 0.5(u_1\sigma)*h + 0.5(u_2\sigma)*h$. Since $u_1(t)=te^{2it}$  and $u_2(t)=te^{-2it}$, we have 
\begin{eqnarray*}
u_1\sigma & = & h_3*h_3, \quad h_3(t) = e^{2it}\sigma(t)\\
u_2\sigma & = & h_4*h_4, \quad h_4(t) = e^{-2it}\sigma(t)
\end{eqnarray*}
Therefore
\begin{eqnarray*}
(u_1\sigma)*h & = & h_1*h_2*h_3*h_3, \quad h_1(t)=h_3(t)=e^{2it}\sigma(t), \quad h_2(t)=e^{-2it}\sigma(t)\\
(u_2\sigma)*h & = & h_1*h_2*h_4*h_4, \quad h_1(t)=e^{2it}\sigma(t), \quad h_2(t)=h_4(t)=e^{-2it}\sigma(t)
\end{eqnarray*}
and then
\begin{eqnarray*}
((u_1\sigma)*h)(t) & = & C_0e^{-2it} + p(t)e^{2it}, \quad p(t) = C_1 + C_2t + C_3t^2/2\\ 
((u_2\sigma)*h)(t) & = & D_0e^{2it} + q(t)e^{-2it}, \quad q(t) = D_1 + D_2t + D_3t^2/2
\end{eqnarray*}
where
\begin{eqnarray*}
\begin{bmatrix} 1 & 1 & 0 & 0\\ -2i & 2i & 1 & 0\\-4 & -4 & 4i & 1\\8i & -8i & -12 & 6i\end{bmatrix}
\begin{bmatrix} C_0 \\ C_1 \\ C_2 \\ C_3\end{bmatrix} & = & \begin{bmatrix}0\\0\\0\\1\end{bmatrix} \implies
\begin{bmatrix} C_0 \\ C_1 \\ C_2 \\ C_3\end{bmatrix} = \begin{bmatrix} - 0.015625i \\ 0.015625i \\ 0.0625 \\ -0.25i\end{bmatrix}\\
\begin{bmatrix} 1 & 1 & 0 & 0\\ 2i & -2i & 1 & 0\\-4 & -4 & -4i & 1\\-8i & 8i & -12 & -6i\end{bmatrix}
\begin{bmatrix} D_0 \\ D_1 \\ D_2 \\ D_3\end{bmatrix} & = & \begin{bmatrix}0\\0\\0\\1\end{bmatrix} \implies
\begin{bmatrix} D_0 \\ D_1 \\ D_2 \\ D_3\end{bmatrix} = \begin{bmatrix} 0.015625i \\ - 0.015625i \\ 0.0625 \\ 0.25i\end{bmatrix}
\end{eqnarray*}
Since $y_p = 0.5(u_1\sigma)*h + 0.5(u_2\sigma)*h$ we have, after regrouping the terms
\[
y_p(t) = -0.03125\sin2t + 0.0625t\cos2t + 0.125t^2\sin2t
\]
and the solution $y=y_h+y_p$ will be given by
\[
y(t) = 1.96875\sin2t -2\cos2t + 0.0625t\cos2t + 0.125t^2\sin2t
\]
\end{enumerate}
\end{example}

\subsection{Difference Equations}\label{exampled}

\begin{example}\em 
Let be the third order initial value problem (IVP):
\begin{equation}\label{dex1}
y(k+3) -1.5y(k+2) + 0.75y(k+1) - 0.125y(k) = 1, \quad y(0)=-1,  y(1)=2, y(2)=0.8
\end{equation}
To find the solution $y$, we consider its characteristic equation $z^3 -1.5z^2+0.75z - 0.125 =0$ 
whose roots as $r_1=r_2=r_3=0.5$.
\begin{enumerate}
\item[(a)] Impulse response ($k\ge 1$):
\(
h(k) = p(k)(0.5)^{k-1}
\)
where $p(k) = A_1 + (k-1)A_2/0.5 + 0.5(k-1)(k-2)A_3/(0.5)^2$ with $A_1$, $A_2$ and $A_3$ being computed as
\[
\begin{bmatrix}1 & 0 & 0\\ 0.5 & 1 & 0\\ 0.25 & 1 & 1\end{bmatrix}\begin{bmatrix}A_1\\ A_2\\ A_3\end{bmatrix} = 
\begin{bmatrix}0 \\ 0 \\ 1\end{bmatrix} \implies 
\begin{bmatrix}A_1\\ A_2\\ A_3\end{bmatrix} = \begin{bmatrix}0 \\ 0 \\ 1\end{bmatrix}
\]
and then 
\[
h(k) = (k-1)(k-2)(0.5)^{k-2}, \quad k\ge 1.
\]

\item[(b)]Homogeneous solution ($k\ge 0$):
\(
y_h(k) = p(k)(0.5)^k
\)
where $p(k) = B_0+kB_1/0.5+0.5k(k-1)B_2/(0.5)^2$ with $B_0$, $B_1$ and $B_2$ being computed as:
\[
\begin{bmatrix}1 & 0 & 0 \\ 0.5 & 1 & 0\\ 0.25 & 1 & 1\end{bmatrix}\begin{bmatrix}B_0\\ B_1\\ B_2\end{bmatrix} = 
\begin{bmatrix}-1 \\ 2\\ 0.8\end{bmatrix} \implies
\begin{bmatrix}B_0\\ B_1\\ B_2\end{bmatrix} = \begin{bmatrix}-1 \\ 2.5\\ -1.45\end{bmatrix}
\]
and then
\[
y_h(k) = -(0.5)^k + 2.5k(0.5)^{k-1} - 1.45k(k-1)(0.5)^{k-1}, \quad k\ge 0.
\]

\item[(c)] Particular solution ($k\ge 0$): $y_p=(u\sigma)*h$, and $(u\sigma)= (1\sigma)= \sigma$ then 
\[
y_p=(u\sigma)*h = h*(u\sigma) = h_1*h_2*h_3*\sigma
\]
where $h_1(k)=h_2(k)=h_3(k) = (0.5)^{k-1}\sigma(k-1)$. Since $\sigma(0)=1$ we will first 
compute $\bar y_p = h_1*h_2*h_3*[\sigma]_1$, where $[\sigma]_1(k)=\sigma(k-1)$, 
in order we have a convolution 
in the format as required in Theorem~\ref{dconfvandn}; at the end we take $y_p(k)=\bar y_p(k+1)$. Then:
\[
\bar y_p(k) = q(k)h_1(k) + C_4\sigma(k-1), \quad q(k) = C_1 + (k-1)C_2/(0.5) + 0.5(k-1)(k-2)C_3/(0.5)^2,
\quad k\ge 1
\]
where $C_1, C_2$, $C_3$ and $C_4$ are compute as the solution of the ``augmented" Vandermonde system:
\[
\begin{bmatrix}1 & 0 & 0 & 1\\ 0.5 & 1 & 0 & 1\\ 0.5^2 & 1 & 1 & 1\\ 0.5^3 & 0.75 & 1.5 & 1\end{bmatrix}
\begin{bmatrix}C_1\\ C_2 \\C_3 \\ C_4\end{bmatrix} = \begin{bmatrix}0 \\ 0 \\ 0 \\1\end{bmatrix} \implies
\begin{bmatrix}C_1\\ C_2 \\C_3 \\ C_4\end{bmatrix} = \begin{bmatrix}-8 \\ -4 \\ -2 \\ 8\end{bmatrix}
\]
Then 
\[
\bar y_p(k) = -8(0.5)^{k-1} - 4(k-1)(0.5)^{k-2} - (k-1)(k-2)(0.5)^{k-3} + 8, \quad k\ge 1
\]
and since $y_p(k) = \bar y_p(k+1)$ we have
\[
y_p(k) = -8(0.5)^{k} - 4k(0.5)^{k-1} - k(k-1)(0.5)^{k-2} + 8, \quad k\ge 0
\]
\end{enumerate}

Finally, the solution for the IVP (\ref{dex1}) is $y=y_h+y_p$, or
\[
y(k) = -9(0.5)^k - 1.5k(0.5)^{k-1} - 3.45k(k-1)(0.5)^{k-1} + 8, \quad k\ge 0
\]
\end{example}

\begin{example}\em 
Let be the third order IVP:
\begin{equation}\label{dex2}
y(k+3) -1.4y(k+2) + 0.9y(k+1) - 0.2y(k) = k, \quad y(0)=2,  y(1)=-3, y(2)=0.5
\end{equation}
whose characteristic  polynomial is $-0.2 + 0.9z^2 - 1.4z^2 + z^3=(z-0.4)[(z-0.5)^2+0.5]$.

\begin{enumerate}
\item[(a)] Impulse response ($k\ge 1$):
\(
h(k) = A_1(0.4)^{k-1} + A_2(1/\sqrt{2})^{k-1}e^{j(k-1)\pi/4} + A_3(1/\sqrt{2})^{k-1}e^{-j(k-1)\pi/4}
\)
where $A_1$, $A_2$ and $A_3$ are computed as
\[
\begin{bmatrix}1 & 1 & 1\\ 0.4 & 0.5+0.5i & 0.5-0.5i\\ 0.16 & 0.5i & -0.5i\end{bmatrix}
\begin{bmatrix}A_1\\ A_2\\ A_3\end{bmatrix} = 
\begin{bmatrix}0 \\ 0 \\ 1\end{bmatrix} \implies 
\begin{bmatrix}A_1\\ A_2\\ A_3\end{bmatrix} = 
\begin{bmatrix}3.8461538\\ -1.9230769 - 0.3846154i\\ - 1.9230769 + 0.3846154i\end{bmatrix}
\]
and then 
\begin{eqnarray*}
h(k) & = & 3.8461538(0.4)^{k-1} + (-1.9230769 - 0.3846154i)(1/\sqrt{2})^{k-1}e^{j(k-1)\pi/4} + \\
        &    & (- 1.9230769 + 0.3846154i)(1/\sqrt{2})^{k-1}e^{-j(k-1)\pi/4}
\end{eqnarray*}
or
\begin{eqnarray*}
h(k) & = & 3.8461538(0.4)^{k-1} -  3.8461538(1/\sqrt{2})^{k-1}\cos[(k-1)\pi/4] + \\
        & & 0.7692308(1/\sqrt{2})^{k-1}\sin[(k-1)\pi/4], \quad k\ge 1
\end{eqnarray*}

\item[(b)]Homogeneous solution ($k\ge 0$):
\(
y_h(k)=B_0(0.4)^{k} + B_1(1/\sqrt{2})^{k}e^{jk\pi/4} + B_2(1/\sqrt{2})^{k}e^{-jk\pi/4}
\)
with $B_0$, $B_1$ and $B_2$ being computed as:
\[
\begin{bmatrix}1 & 1 & 1\\ 0.4 & 0.5+0.5i & 0.5-0.5i\\ 0.16 & 0.5i & -0.5i\end{bmatrix}
\begin{bmatrix}B_0\\ B_1\\ B_2\end{bmatrix} = 
\begin{bmatrix}2\\ -3 \\ 0.5\end{bmatrix} \implies 
\begin{bmatrix}B_0\\ B_1\\ B_2\end{bmatrix} = 
\begin{bmatrix} 17.307692 \\ - 7.6538462 + 2.2692308i \\ - 7.6538462 - 2.2692308i\end{bmatrix}
\]
and then
\[
y_h(k) = 17.307692(0.4)^k - 15.307692(1/\sqrt{2})^k\cos(k\pi/4) - 4.5384615(1/\sqrt{2})^k\sin(k\pi/4)
\]

\item[(c)] Particular solution ($k\ge 0$): $y_p=(u\sigma)*h$ where $u(k)=k$, and then we need to write 
$k\sigma(k)$ as a sum of convolution of signals. From Remark~\ref{dfoldconvm} we have that $(\sigma*\sigma)(k)=k+1$,
then we easily get
\[
k\sigma(k) = (\sigma*\sigma)(k) - \sigma(k), \quad\text{or } (u\sigma)=(\sigma*\sigma) - \sigma
\]
then 
\[
y_p=(u\sigma)*h = h*(u\sigma) = \underbrace{h_1*h_2*h_3*\sigma*\sigma}_{y_{p_1}} - 
\underbrace{h_1*h_2*h_3*\sigma}_{y_{p_2}}
\]
where $h_1(k)=(0.4)^{k-1}\sigma(k-1), h_2(k)=(0.5+0.5i)^{k-1}\sigma(k-1)$ and 
$h_3(k) = (0.5-0.5i)^{k-1}\sigma(k-1)$ and $y_{p_1}$ and $y_{p_2}$ can be calculated as:

{\bf(c.1)} $y_{p_1}(k) = \bar y_{p_1}(k+2)$, where $\bar y_{p_1}= h_1*h_2*h_3*[\sigma]_1*[\sigma]_1$, 
or
\[
\bar y_{p_1}(k) = C_1(0.4)^{k-1} + C_2(1/\sqrt{2})^{k-1}e^{j(k-1)\pi/4} + C_3(1/\sqrt{2})^{k-1}e^{-j(k-1)\pi/4} +
C_4 + C_5(k-1),
\]
with $C_1,C_2,C_3,C_4$ and $C_5$ computed by solving
\[
\begin{bmatrix}1 & 1 & 1 & 1 & 0\\ 0.4 & 0.5+0.5i & 0.5-0.5i & 1 & 1\\ 0.16 & 0.5i & -0.5i & 1 & 2\\
0.064 &  - 0.25 + 0.25i & - 0.25 - 0.25i   & 1   & 3 \\ 0.0256  & - 0.25 & - 0.25 & 1 & 4\end{bmatrix}
\begin{bmatrix}C_1\\ C_2\\ C_3 \\ C_4 \\ C_5\end{bmatrix} = 
\begin{bmatrix}0\\ 0 \\ 0 \\ 0 \\ 1\end{bmatrix} \implies 
\begin{bmatrix}C_1\\ C_2\\ C_3 \\ C_4 \\ C_5\end{bmatrix} = 
\begin{bmatrix}  10.683761 \\ 0.7692308 - 3.8461538i \\ 0.7692308 + 3.8461538i \\ - 12.222222 \\ 3.3333333\end{bmatrix}
\]
and so
\begin{eqnarray*}
\bar y_{p_1}(k) & = & 10.683761(0.4)^{k-1} + 1.5384615(1/\sqrt{2})^{k-1}\cos[(k-1)\pi/4] + \\
                         &    & 7.6923077(1/\sqrt{2})^{k-1}\sin[(k-1)\pi/4] - 12.222222 + 3.3333333(k-1)
\end{eqnarray*}
then $y_{p_1}(k)=\bar y_{p_1}(k+2)$ is given by
\begin{eqnarray*}
y_{p_1}(k) & = & 10.683761(0.4)^{k+1} + 1.5384615(1/\sqrt{2})^{k+1}\cos[(k+1)\pi/4] + \\
                         &    & 7.6923077(1/\sqrt{2})^{k+1}\sin[(k+1)\pi/4] - 12.222222 + 3.3333333(k+1)
\end{eqnarray*}
{\bf(c.2)} $y_{p_2}(k) = \bar y_{p_2}(k+1)$, where $\bar y_{p_2}=h_1*h_2*h_3*[\sigma]_1$, or
\[
\bar y_{p_2}(k) = D_1(0.4)^{k-1} + D_2(1/\sqrt{2})^{k-1}e^{j(k-1)\pi/4} + D_3(1/\sqrt{2})^{k-1}e^{-j(k-1)\pi/4} +
D_4,
\]
with $D_1,D_2,D_3$ and $D_4$ computed by solving
\[
\begin{bmatrix}1 & 1 & 1 & 1 \\ 0.4 & 0.5+0.5i & 0.5-0.5i & 1\\ 0.16 & 0.5i & -0.5i & 1\\
0.064 &  - 0.25 + 0.25i & - 0.25 - 0.25i   & 1\end{bmatrix}
\begin{bmatrix}D_1\\ D_2\\ D_3 \\ D_4\end{bmatrix} = 
\begin{bmatrix}0\\ 0 \\ 0 \\ 1\end{bmatrix} \implies 
\begin{bmatrix}D_1\\ D_2\\ D_3 \\ D_4\end{bmatrix} = 
\begin{bmatrix}  - 6.4102564 \\ 1.5384615 + 2.3076923i  \\ 1.5384615 - 2.3076923i \\ 3.3333333\end{bmatrix}
\]
which implies
\begin{eqnarray*}
\bar y_{p_2}(k) & =  & - 6.4102564(0.4)^{k-1} + 3.0769231(1/\sqrt{2})^{k-1}\cos[(k-1)\pi/4] - \\
&& 4.6153846(1/\sqrt{2})^{k-1}\sin[(k-1)\pi/4] + 3.3333333\\
\end{eqnarray*}
and then
\[
y_{p_2}(k) =  - 6.4102564(0.4)^{k} + 3.0769231(1/\sqrt{2})^{k}\cos[k\pi/4] - 
 4.6153846(1/\sqrt{2})^{k}\sin[k\pi/4] + 3.3333333
\]
Therefore $y_p=y_{p_1} + y_{p_2}$ is given by
\begin{eqnarray*}
y_p(k) & = & 10.683761(0.4)^{k+1} + 1.5384615(1/\sqrt{2})^{k+1}\cos[(k+1)\pi/4] + \\
&& 7.6923077(1/\sqrt{2})^{k+1}\sin[(k+1)\pi/4] - 12.222222 + 3.3333333(k+1) - 6.4102564(0.4)^{k} + \\
&& 3.0769231(1/\sqrt{2})^{k}\cos[k\pi/4] - 4.6153846(1/\sqrt{2})^{k}\sin[k\pi/4] +  3.3333333
\end{eqnarray*}

\end{enumerate}
Finally, the solution for the IVP (\ref{dex1}) is $y=y_h+y_p$, or
\begin{eqnarray*}
y(k) & = & 17.307692(0.4)^k - 15.307692(1/\sqrt{2})^k\cos(k\pi/4) - 4.5384615(1/\sqrt{2})^k\sin(k\pi/4) +\\
& &10.683761(0.4)^{k+1} + 1.5384615(1/\sqrt{2})^{k+1}\cos[(k+1)\pi/4] + \\
&& 7.6923077(1/\sqrt{2})^{k+1}\sin[(k+1)\pi/4] - 12.2222222 + 3.3333333(k+1) - \\
&& 6.4102564(0.4)^{k} + 3.0769231(1/\sqrt{2})^{k}\cos[k\pi/4] -  4.6153846(1/\sqrt{2})^{k}\sin[k\pi/4] + \\
&& 3.3333333
\end{eqnarray*}
which, in turn, can be simplified to
\[
y(k) = 15.17094(0.4)^k - 7.6153843(1/\sqrt{2})^k\cos(k\pi/4) - 6.076923(1/\sqrt{2})^k\sin(k\pi/4) + 
3.3333333k - 5.5555556
\]
\end{example}

\begin{example}\em 
Let be the third order IVP:
\begin{equation}\label{dex3}
y(k+3) + y(k+1) = \sin(k\pi/2), \quad y(0)=1,  y(1)=y(2)=0
\end{equation}
whose characteristic equation is $z^3+z=z(z^2+1)=0$, and $r_1=j, r_2=-j$ and $r_3=0$. 
We discard $r_3=0$ and solve a second order equation and at the end shift the solution by
one unity to the right.
\end{example}

\begin{enumerate}
\item[(a)] Impulse response ($k\ge 1$):
\(
h(k) = A_1(j)^{k-1} + A_2(-j)^{k-1}
\)
where $A_1$ and $A_2$ are computed as
\[
\begin{bmatrix}1 & 1\\ j & -j\end{bmatrix}
\begin{bmatrix}A_1\\ A_2\end{bmatrix} = 
\begin{bmatrix}0\\ 1\end{bmatrix} \implies 
\begin{bmatrix}A_1\\ A_2\end{bmatrix} = 
\begin{bmatrix} -j/2\\ j/2 \end{bmatrix}
\]
which implies $h(k)=-j/2(j)^{k-1} + j/2(-j)^{k-1}$ or $h(k) = -\cos(k\pi/2)$, $k\ge 1$
\item[(b)] Homogeneous solution ($k\ge 0$):
\(
y_h(k)=B_0(j)^{k} + B_1(-j)^{k}
\)
with $B_0$ and $B_1$ being computed as:
\[
\begin{bmatrix}1 & 1 \\ j & -j\end{bmatrix}
\begin{bmatrix}B_0\\ B_1\end{bmatrix} = 
\begin{bmatrix}0\\ 0\end{bmatrix} \implies 
\begin{bmatrix}B_0\\ B_1\end{bmatrix} = 
\begin{bmatrix} 0\\0\end{bmatrix}
\]
and then $y_h(k)=0$.

\item[(c)] Particular solution ($k\ge 0$): $y_p=(u\sigma)*h$ where $u(k)=\sin(k\pi/2) = (1/2j)(j)^k - (1/2j)(-j)^k$.
Then 
\[
y_p = (1/2j)\underbrace{(u_1*h_1*h_2)}_{y_{p_1}} - (1/2j)\underbrace{(u_2*h_1*h_2)}_{y_{p_2}}
\]
where $u_1(k)=(j)^k\sigma(k), u_2(k)=(-j)^k\sigma(k), h_1(k)=(j)^{k-1}\sigma(k-1), h_2(k)=(-j)^{k-1}\sigma(k-1)$
and $y_{p_1}(k)=\bar y_{p_1}(k+1)$ and $y_{p_2}(k)=\bar y_{p_2}(k+1)$:

{\bf(c.1)} $\bar y_{p_1}(k) = p(k)h_1 + C_3h_2(k)$, where $p(k) = C_1+(k-1)C_2/j$, and 
\[
\begin{bmatrix}1 & 0 & 1\\ j & 1 & -j\\ j^2 & 2j & (-j)^2\end{bmatrix}
\begin{bmatrix}C_1\\ C_2 \\C_3\end{bmatrix} = \begin{bmatrix}0 \\ 0\\1\end{bmatrix} \implies
\begin{bmatrix}C_1\\ C_2 \\C_3\end{bmatrix} = \begin{bmatrix}1/4 \\ -j/2\\ -1/4\end{bmatrix}
\]
Then $\bar y_{p_1}(k) = 1/4(j)^{k-1} - 1/2(k-1)(j)^{k-1} - 1/4(-j)^{k-1}$ and so
\[
y_{p_1}(k) = \bar y_{p_1}(k+1) = 1/4(j)^k - 1/2k(j)^k - 1/4(-j)^k, \quad k\ge 0
\]

{\bf(c.2)} $\bar y_{p_2}(k) = D_1h_1 + q(k)h_2(k)$, where $q(k) = D_2+(k-1)D_3/(-j)$, and 
\[
\begin{bmatrix}1 & 1& 0\\ j & -j & 1\\ j^2 & (-j)^2 & -2j\end{bmatrix}
\begin{bmatrix}D_1\\ D_2 \\D_3\end{bmatrix} = \begin{bmatrix}0 \\ 0\\1\end{bmatrix} \implies
\begin{bmatrix}D_1\\ D_2 \\D_3\end{bmatrix} = \begin{bmatrix}-1/4 \\ 1/4\\ j/2\end{bmatrix}
\]
Then $\bar y_{p_2}(k) = -1/4(j)^{k-1} + 1/4(-j)^{k-1} - 1/2(k-1)(-j)^{k-1}$ and so
\[
y_{p_2}(k) = \bar y_{p_2}(k+1) = -1/4(j)^k + 1/4(-j)^k - 1/2k(-j)^k, \quad k\ge 0
\]

Now, since $y_p = (1/2j)y_{p_1} - (1/2j)y_{p_2}$, we have
\[
y_p(k) = (1/2j)[1/4(j)^k - 1/2k(j)^k - 1/4(-j)^k] - (1/2j)[-1/4(j)^k + 1/4(-j)^k - 1/2k(-j)^k]
\]
which can (optionally) be simplified to $y_p(k) = (1/2)\sin(k\pi/2)-(k/2)\sin(k\pi/2)$. Finally,
to contemplate the zero root of the characteristic equation and the initial condition $y(0)=1$,
we have:
\[
y(k) = 2\delta(k) + (1/2)\sin[(k-1)\pi/2] - [(k-1)/2]\sin[(k-1)\pi/2], \quad k\ge0
\]
or
\[
y(k) = 2\delta(k) + (k/2)\cos(k\pi/2)-\cos(k\pi/2), \quad k\ge 0
\]
\end{enumerate}

\section{Conclusions}

We showed in this paper a technique for computing the convolution of exponential signals, in analog 
and discrete time context, that avoids the resolution of integrals and summations. The method is 
essentially algebraic and requires the resolution of Vandermonde systems, which is a well-known 
and extensively discussed problem in literature (see e.g. \cite{golub,houwan} and references therein). 
While the question of computing convolution of exponentials have been discussed previously in 
literature (\cite{akk,maliu}), the proposed approach is apparently different from the previous ones,
and additionally is quite simple and suitable to be implemented computationally. Finally, we use the 
proposed approach to solve a $n$ order differential/difference equation with constant coefficients.

\bibliographystyle{plain}

\end{document}